\documentclass[9pt,technote]{IEEEtran}
\usepackage{graphicx,amssymb,arydshln,amsmath,amssymb,amsfonts}

\newcommand{\cZ}{{\cal Z}}
\newcommand{\Ht}{{\sf H^2}}
\newcommand{\sss}{\scriptscriptstyle}
\newcommand{\scs}{\scriptstyle}
\newcommand{\sN}{{\scriptstyle N}}
\newcommand{\Htpq}{{\sf H^2_{\raisebox{.1em}[0em]{$\sss p\times q$}}}}
\newcommand{\speq}{\,=\,}
\newcommand{\req}[1]{(\ref{#1.eq})} 
\newcommand{\mycaption}[1]{\caption{\footnotesize #1}}
\newcommand{\R}{\mathbb{R}}
\newcommand{\sm}{\text{-}}

\newtheorem{theorem}{Theorem}
\newtheorem{lemma}[theorem]{Lemma}

\begin{document}

\title{On the connection between input-output resonances and internal
    modes of linear time-invariant systems}

    \author{Bassam Bamieh\thanks{
    Department of Mechanical Engineering,
    University of California at Santa Barbara,
    Santa Barbara, CA 93106.
    {\em bamieh@ucsb.edu}}}

\maketitle

\begin{abstract}
It is shown  that in general, there is no connection between the
location of the internal modes of a Linear Time-Invariant (LTI)
system and the shape of its input-output frequency response. In
particular, it is shown  that resonance peaks of the frequency response
do not necessarily correspond to under-damped internal modes. 
This phenomenon, though rare, can occur in high (or
infinite) dimensional LTI systems. In the Single Input Single
Output (SISO) case, this phenomenon can be attributed to the
location of system zeros, while in certain Multi Input Multi
Output (MIMO) cases without system zeros, it can be attributed to
the non-normality of the matrix generating the internal dynamics.
\end{abstract}

\section{Introduction}

In input-output analysis of Linear Time Invariant (LTI) systems,
the frequency response plays a dominant role. The peaks of the
frequency response indicate frequencies that are most ``resonant''
by the system's dynamics, i.e. inputs at that frequency are most
amplified in the system's output. On the other hand, for systems
with or without inputs, the internal modes of the system are
usually thought of as representing the ``natural modes'' of the
dynamics.

More precisely, consider  the finite dimensional stable LTI
system be given by the following minimal realization
\begin{align*} 
	 \dot{x} & \speq   Ax +  B u    \\
                 y   & \speq   Cx,
 \end{align*}
where $A$ has all its eigenvalues in the open Left Half of the complex
Plane (LHP).  As is well known,
the frequency response of this system is given by
    \[
        G(j\omega) ~=~ C(j\omega I-A)^{-1} B.
    \]
When the matrix $A$ has eigenvalues with small negative real
part (commonly referred to as \textit{under-damped modes}), one
expects that the corresponding frequency response has peaks at
frequencies near (the imaginary parts of) these internal modes.
The reasoning for this is due to the fact that at the eigenvalues
of $A$, the transfer function $G(s)$ has infinite magnitude. These
observations are commonly stated in classical vibrations and controls 
textbooks~\cite{ogata2010modern,Tongue2001}.
For Single Input Single Output systems (SISO) systems, it is easy
to see that this picture has to be modified due to the presence of
transfer function zeros which may modify the effect of nearby
poles. Figure~\ref{rubber_sheet.fig} gives an illustration. 

A further modification to the  simple picture in Figure~\ref{rubber_sheet.fig}
 occurs if a
certain under-damped internal mode is weakly controllable and/or
observable from inputs and outputs respectively. In this case, it
may not effect the frequency response greatly.

For Multi Input Multi Output (MIMO) systems, the relations are  
more complicated than those described above. However, certain 
cases can be easily understood. Consider the case of a MIMO 
transfer function of the form $G(s) = C(sI-A)^{-1} B$ where $C,B$ are square
unitary matrices, and $A$ is a normal matrix (i.e. it has mutually orthogonal 
eigenvectors). For MIMO systems the frequency response is
quantified by the singular value plot 
\begin{multline} 
	\sigma_{\max} \left(  C (j\omega I -A)^{-1} B \right) 
	=^1 \sigma_{\max} \left(   (j\omega I -A)^{-1}  \right) 		\\
	= \frac{1}{ \sigma_{\min}  (j\omega I -A) }    
	=^2 \frac{1}{\min_\ell   ~| j\omega  -\lambda_\ell(A)|}  , 	\label{norm_dist.eq}
\end{multline} 
where $=^1$ follows from $C,B$ being unitary, $\lambda_i(A)$ is the 
$i$'th eigenvalue of $A$, and $=^2$ follows from $A$ being 
normal. The quantity~\req{norm_dist} has a simple geometric 
interpretation since $|j\omega - \lambda_\ell(A) |$  
 is the distance in the complex plane between $j\omega$ and $\lambda_\ell(A)$. Thus 
the singular value plot at $j\omega$  is inversely proportional to the distance 
between $j\omega$ and the closest eigenvalue of $A$ as 
illustrated in Figure~\ref{MIMO_Norm.fig}. This indicates again that 
poles close to the imaginary axis induce resonance peaks at close by 
frequencies  for such 
transfer functions.

\begin{figure}[t]
	\centering
	\includegraphics[width=.35\textwidth]{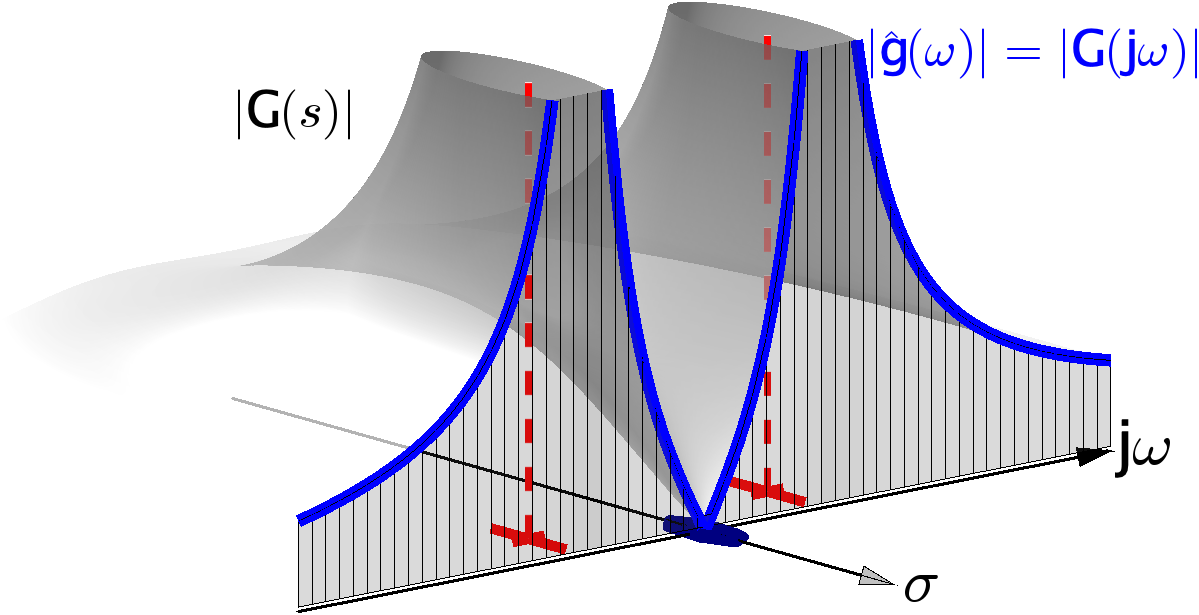} 
	
	\mycaption{Typical relation between the frequency response of a SISO LTI system and its 
		pole locations. Underdamped poles correspond to resonances which 
		are frequency response  peaks. 
		} 
  \label{rubber_sheet.fig}
\end{figure}

\begin{figure}[t]
	\centering
	\includegraphics[width=.3\textwidth]{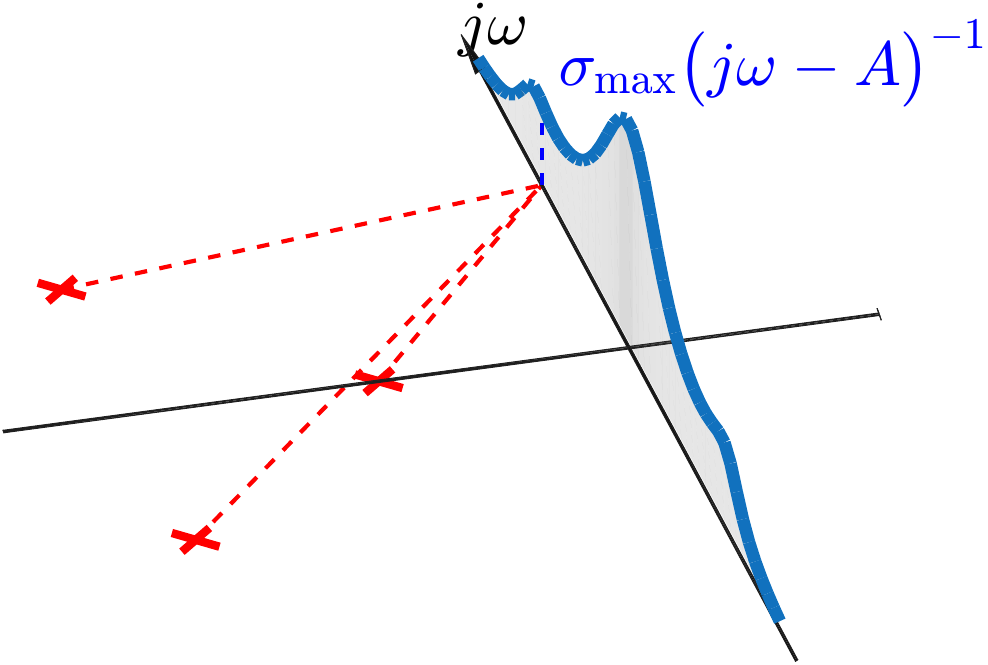} 
	
	\mycaption{For the MIMO transfer function $(sI-A)^{-1}$ where $A$ is normal, 
		the singular value plot at $j\omega$ is inversely proportional to the distance between 
		$j\omega$ and the closest pole location. 
		} 
  \label{MIMO_Norm.fig}
\end{figure}

The argument in~\req{norm_dist} fails if $A$ is nonnormal. This indicates that for such 
systems, the relations between underdamped modes and resonances 
may not be as simple. Thus it may be possible that a lightly damped 
mode does not correspond to a frequency response peak. It is perhaps even 
more surprising that a resonance peak can appear in the frequency response 
curve without there being a lightly damped mode nearby. 

In this paper, we show that in general, there need not be any
connection between the locations of internal modes in the LHP and
the peaks of the corresponding frequency response. More precisely,
we show that given any a priori specified pole locations in the
LHP, one can approximate any frequency response (of a transfer
function in the Hardy space $\Ht$) arbitrarily closely with a
transfer function having its poles at the pre-specified locations.
%We show this using several constructions that demonstrate several
%possible reasons for this phenomenon. First, we present a SISO
%construction in which the lack of correspondence between internal
%modes and input-output resonances can be attributed to the action
%of system zeros. Then we give a construction of a MIMO system
%without transmission zeros, in which this lack of correspondence
%can be attributed to the non-normality (non-orthogonality of
%eigenvectors) of the generator $A$. 
We show this using a construction that typically requires high-order 
systems. Thus this unusual lack of correspondence between internal 
modes and external resonance is to be understood as a high-dimensional 
phenomenon. 
%systems are necessary for these constructions, and we present
%cases in which systems have either repeated or non-repeated poles.

Though this lack of correspondence between under-damped internal
modes and input-output resonances is atypical, it does occur in
some important systems. A notable such case is input-output 
analysis of wall-bounded shear flow transition to turbulence~\cite{jovanovic2005componentwise}, 
in which a long-standing anomaly in hydrodynamic stability analysis 
 can be attributed to this lack of
correspondence.

\section{Main Result} 

\subsection{Mathematical Preliminaries}

We carry out the approximation for transfer functions in the Hardy
space $\Ht$, which is the subspace of $\sf L^2(j\R)$ (on the imaginary
axis) of functions analytic in the right half of the complex
plane. Another characterization of $\Ht$ is as the closure (in the 
$\sf L^2 (j\R)$ norm) of stable rational 
transfer functions. Thus, finite dimensional stable transfer functions
belong to this class. To carry out MIMO analysis, we allow those
transfer function to be matrix valued. We will use the notation
$\Htpq$ to refer to the $\Ht$ space of  $p\times q$ 
matrix-valued functions.

\subsection{Main Result}

\begin{lemma}\em							\label{main.lemma}
Let $G(j\omega)$ be the restriction to the imaginary axis of any function 
in $\Htpq$. Let $n\geq 1$ be any integer, 
and let $\cZ :=\{z_1,\ldots, z_n\}$ be any  set of points  in the LHP. 
Given any $\epsilon>0$, there exists integers $\{\sN_1,\ldots,\sN_n\}$, and a set of 
$p\times q$ matrices $\{ H_{\ell k} \}$  such that
$G$ can be approximated (in the $\Ht$ norm) to within $\epsilon$ by a transfer
function of the form
\[
	\sum_{\ell=1}^n \sum_{k=1}^{N_\ell} \frac{1}{(s-z_\ell)^{k}} ~H_{\ell k} .
\]
\end{lemma}

The proof of  this lemma  is a slight generalization of a well-known
technique~\cite{rud62}.

\begin{proof}
Consider the set of functions 
\[
	{\cal E} := 
	 \left\{ \frac{1}{(s-z_\ell)^{k}} \, E_{ij}, ~
	 	  % \right\}_{	\begin{array}{l}	\scriptstyle
	 	k\geq 1,   \begin{array}{c} \scs 1\leq i\leq p, \\ \scs  1\leq j\leq q, \end{array} 
			~ 1 \leq \ell \leq n  \right\} 
%									\end{array}}, \ldots, 
%	 \left\{ \frac{E_{ij}}{(s-z_n)^{k}} \right\}_{	\begin{array}{l}	\scriptstyle
%	 								k\geq 1 \\ 		\scriptstyle
%	 								1\leq i\leq p, ~1\leq j\leq q
%									\end{array}},
\]
where $E_{ij}$ are the elementary $p\times q$ matrices with $1$ in the
$i,j$'th position, and zeros elsewhere. The Lemma follows from showing 
that this set spans a dense subspace in $\Htpq$. 
This is accomplished by showing that any element 
$G\in\Htpq$ perpendicular to the span of the above functions must be zero. 
Note that the inner product in $\Htpq$ between $G$ and any of 
the above functions is
									%\marginpar{check this}
\begin{align*}
	\left<  \frac{E_{ij}}{(s-z_l)^{k}} ~ , G \right>  
	& =  \tfrac{1}{2\pi} 
	\int_{\sm \infty}^\infty
		{\sf trace} \!	\left(
		\left( \frac{E_{ij}}{(j\omega-z_l)^k}\right)^*  G(j\omega)
				\right)  \, d\omega
												\\
	& =  
		j~\frac{d^{(k-1)}G_{ij}}{dz^{(k-1)}}(-z^*_l),		
\end{align*}
where the second equality follows from Cauchy's integral formula. 			
Now if $G$ is orthogonal to all of the above basis functions, then each entry of the matrix $G$
and all its derivatives are identically zero at each points of $-\cZ$. 
Since the points of $-\cZ$ are in the RHP,  which is  the region of analyticity of $G$, then each entry of $G$ is 
identically zero. 
\end{proof}

\section{Discussion} 

The lemma above is illustrated in Figure~\ref{modal_io_NOconnection2.fig}. 
%The result is not that
%surprising from a system theoretic point of view, especially when one thinks about 
%approximate controllability and observability of internal modes. 
The construction 
in the proof of Lemma~\ref{main.lemma} will typically require high pole multiplicities. Alternatively 
each pole with multiplicity can be approximated with a high order system without pole 
multiplicity, but with an aggregation of poles in a cluster around that location. In any of 
those approximations, high-order approximants would be required. Thus the phenomenon 
of contrast between internal modes and external resonances explored in this paper 
is fundamentally a high-dimensional phenomenon.

As already mentioned, 
it is rather rare to encounter such systems. However, this is precisely the case for 
spatio-temporal transfer functions representing the linearized Navier-Stokes equations for 
wall-bounded shear flows~\cite{jovanovic2005componentwise}. 
The particular case of channel flows has a very lightly damped pole (compared to all the other poles)
 with non-zero imaginary part ($\omega \neq 0$)
known as Tollmien-Schlichting waves. 
Yet, very large  input-output (spatio-temporal) resonances occur near $\omega =0$
without the presence of any lightly damped poles nearby. From an input-output point of 
view, these resonances are much more physically relevant than the lightly damped poles 
with very little effect on the external dynamic behavior. The flow structures corresponding to 
those external resonances are typically the ones seen in experiments and numerical studies of 
transition and turbulence in the so-called bypass transition regime. 
Thus the lack of correspondence between internal modes and external resonances 
pointed out in this paper is  one explanation for a  
long-standing paradox in hydrodynamic stability studies of wall-bounded shear 
flows~\cite{schmid2007nonmodal}. 

\begin{figure}[t]
	\centering
	\includegraphics[width=.3\textwidth]{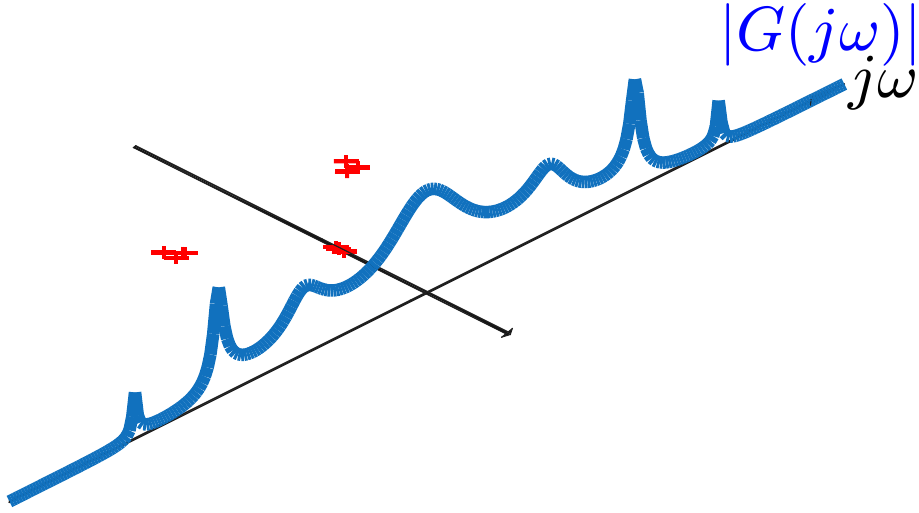} 
	
	\mycaption{Illustration of the main result Lemma~\ref{main.lemma}. 
		Given any set of points $\cZ$ in the LHP (red crosses), any frequency response in 
		$\sf H^2$ can be arbitrarily closely approximated by 
		a  transfer function whose poles are restricted to $\cZ$. 
		} 
  \label{modal_io_NOconnection2.fig}
\end{figure}

\bibliographystyle{ieeetr}
\bibliography{io_res}

@article{schmid2007nonmodal,
	author = {Schmid, Peter J},
	journal = {Annu. Rev. Fluid Mech.},
	number = {1},
	pages = {129--162},
	publisher = {Annual Reviews},
	title = {Nonmodal stability theory},
	volume = {39},
	year = {2007}}

@book{Tongue2001,
	address = {New York},
	author = {Benson H. Tongue},
	edition = {2},
	isbn = {978-0-19-514246-4},
	publisher = {Oxford University Press},
	title = {Principles of Vibration},
	year = {2001}}

@book{ogata2010modern,
	author = {Ogata, Katsuhiko},
	publisher = {Prentice hall},
	title = {Modern control engineering},
	year = {2010}}

@article{jovanovic2005componentwise,
	author = {Jovanovi{\'c}, Mihailo R and Bamieh, Bassam},
	journal = {Journal of Fluid Mechanics},
	pages = {145--183},
	publisher = {Cambridge University Press},
	title = {Componentwise energy amplification in channel flows},
	volume = {534},
	year = {2005}}

@book{rud62,
	address = {New York},
	author = {W. Rudin},
	publisher = {Interscience-Wiley},
	title = {Fourier Analysis on Groups},
	year = 1962}

\end{document}